\documentclass{article}

\usepackage[reset,a4paper,vmargin=3.5truecm,hmargin=3truecm]{geometry}
\usepackage{amsmath,amsthm,amssymb,amscd,epic}
\usepackage{graphicx,bbm,amsthm,graphicx}
\usepackage{mathtools}
\usepackage{leftidx}
\usepackage{hyperref}
\usepackage{cool}
\usepackage{bm}
\usepackage{bbold}
\usepackage{subfig}
\usepackage{tikz}
\usepackage{mathrsfs}

\theoremstyle{plain}
\newtheorem{thm}{Theorem}[section]

\newtheorem{prop}[thm]{Proposition}

\newtheorem{assump}[thm]{Assumption}

\theoremstyle{definition}

\newtheorem{ex}{Example}[section]

\theoremstyle{remark}
\newtheorem{rem}{Remark}[section]

\newcommand{\HRabi}[1]{H_{#1}} 

\newcommand{\Z}{\mathbb{Z}} 
\newcommand{\Q}{\mathbb{Q}} 
\newcommand{\R}{\mathbb{R}} 
\newcommand{\C}{\mathbb{C}} 

\newcommand{\PRabi}{Z_{\text{R}}}
\newcommand{\PARabi}[1]{Z^{(#1)}_{\text{R}}}
\newcommand{\ZARabi}[1]{\zeta^{(#1)}_{\text{R}}}
\newcommand{\ZRabi}{\zeta_{\text{R}}}
\newcommand{\KRabi}[1]{K_{\text{R}}^{(#1)}}
\newcommand{\KBRabi}[1]{\overline{K}_{\text{R}}^{(#1)}}
\newcommand{\KBRabip}{\overline{K}_{\text{R}}}

\DeclareMathOperator{\sh}{sh}
\DeclareMathOperator{\ch}{ch}
\DeclareMathOperator{\Spec}{Spec}

\newcommand{\e}{\varepsilon}

\title{Zeta limits for the spectrum of quantum Rabi models}
\author{Cid Reyes-Bustos and Masato Wakayama}

\begin{document}

\maketitle


\begin{abstract}

  The quantum Rabi model (QRM), one of the fundamental models used to describe light and matter interaction, 
  has a deep mathematical structure revealed by the study of its spectrum. In this paper,  from the explicit
  formulas for the partition function we directly derive various limits of the spectral zeta function with respect to the systems parameters of the asymmetric quantum Rabi model (AQRM), a generalization obtained by adding a physically significant parameter to the QRM. In particular, we consider the limit corresponding to the growth of the coupling strength to infinity, recently studied using resolvent analysis. The limits obtained in this paper are given
  in terms of the Hurwitz zeta function and other $L$-functions, suggesting further relations between spectral zeta
  function of quantum interaction models and number theory.

\,
 
 \smallskip

\end{abstract}

\section{Introduction}
\label{sec:introduction}

The quantum Rabi model (QRM) is widely recognized as one of the most fundamental models in quantum optics \cite{JC1963,B2011}. Partly because of its applications for quantum information technologies, it has been in the research spotlight in experimental and theoretical physics in recent years \cite{bcbs2016}. In particular, due to advances in experimental technology, parameters regimes have been achieved in which approximations using the Jaynes-Cummings model \cite{JC1963}, a simpler model with an explicit description of the spectrum, is no longer valid, thus making the study of the QRM and its spectrum essential. Recently, the mathematical aspects underlying its spectrum have also been studied from different points of view, including stochastic analysis (Feynman-Kac's integrals \cite{HH2012, HHL2014}), asymptotic analysis (large eigenvalue asymptotics \cite{BZ2021} and the Weyl law \cite{Sugi2016,RW2021}, see also Braak conjecture on eigenvalue distribution \cite{B2011}), representation theory (generalized orthogonal polynomials, confluent Heun ODE, $\mathfrak{sl}_2$-picture \cite{W2016JPA, KRW2017}, the infinite symmetric group \cite{RW2019}) and number theory (spectral zeta functions and their special values \cite{Sugi2016,RW2021}, geometry of hyperelliptic curves and surfaces \cite{RBW2022}).

A generalization of the QRM of considerable mathematical interest, despite the apparent simplicity of its definition, is the asymmetric quantum Rabi model (AQRM) \cite{B2011}. The Hamiltonian of the AQRM is given by
\[ 
 \HRabi{\e} := \omega a^{\dagger}a + \Delta \sigma_z + g (a + a^{\dagger}) \sigma_x + \e \sigma_x,
\]
where $a^\dag$ and $a$ are the creation and annihilation operators of the bosonic mode,
i.e. $[a,\,a^\dag]=1$ and
\[
\sigma_x = \begin{bmatrix}
 0 & 1  \\
 1 & 0
\end{bmatrix}, \qquad \qquad
\sigma_z= \begin{bmatrix}
 1 & 0  \\
 0 & -1
\end{bmatrix}
\]
are the Pauli matrices, $2\Delta$ is the energy difference between the two levels,
$g$ denotes the coupling strength between the two-level system and the bosonic mode with frequency $\omega$
(subsequently, we set $\omega=1$ without loss of generality) and \(\e\) is a real number. The Hamiltonian of the QRM is 
given by $\HRabi{0}$, that is, it corresponds to the case $\e=0$.

The term $\e \sigma_x$ in $  \HRabi{\e}$ corresponds to spontaneous flips of the two-level system and appears naturally
in implementations with flux qubits \cite{Ni2010} (see also the recent experimental study \cite{Y2017} toward the realization of qubits using superconductivity). Notably, this term breaks the $\Z_2(=\Z/2\Z)$-symmetry of the Hamiltonian of QRM. Prior to the study in \cite{LB2015JPA} it was believed that the AQRM did not posses apparent symmetries, and that no level crossings should be expected in the spectral curve for any $\e \not=0$ (instead, see \cite{K1985JMP} for the case of the QRM). However, level crossings at the quasi-exact (or Juddian) solution  were observed when $\e =\frac12$ in \cite{LB2015JPA} and later proved in \cite{W2016JPA}. Now, it is known \cite{KRW2017} that when $\e $ is a half integer the AQRM has a quasi-exact (exceptional) spectrum that turns to be doubly degenerate (energy level crossing) just like the QRM. Conversely, the full spectrum is simply when $\e \notin \tfrac12 \Z$. 

A relevant problem is the description of the behavior of the spectrum of the AQRM when the value of the coupling
$g$ becomes large (beyond ``deep strong'' coupling regimes \cite{Y2017,YS2018}). The problem has been
considered in physics \cite{Sch1985} by giving a description of the expect ``natural'' behavior, but a 
mathematical proof was not provided.
In the workshop ``Rabi and Spin Boson models'' held in Kyushu University, Fumio Hiroshima presented a study on
the behaviour of the quantum Rabi model in the weak limit $g \to \infty$ by a study of the resolvent \cite{H2023}. The main result is that, in the sense of the resolvent, the spectrum of the QRM Hamiltonian shifted by $g^2$, that is,
\[
 \hat{H} = \HRabi{0} + g^2
\]
converges to that of a pair of quantum harmonic oscillators (QHO). In the case of the AQRM the spectrum
of $\HRabi{\e} + g^2$ in the weak limit is essentially given by that of  the spectrum of the pair of displaced QHO, 
that is,
\[
  a^\dag a + \e, \qquad \qquad  a^\dag a - \e.
\]
In particular, we observe that the limiting eigenvalue distribution for $g\to\infty$ depends only on the flip parameter  $\e$ and not on the
parameter $\Delta$. 

The main purpose of this paper is to give an analogous result in terms of the spectral zeta function of the model using the explicit
formulas for the partition function and the heat kernel in \cite{RW2019,RW2021,R2020}. From the explicit formulas, the limit is computed in an elementary way, highlighting the good convergence properties of the series which describes the partition function and heat kernel formulas. 
The partition function $\PARabi{\e}(\beta;g,\Delta)$ of the AQRM is given by the trace of the Boltzmann factor $e^{-\beta E(\mu)}$,
$E(\mu)$ being the energy of a state $\mu$, that is, 
\[
  \PARabi{\e}(\beta;g,\Delta) := \text{Tr}[e^{-t \HRabi{\e}}]= \sum_{\mu\in \Lambda} e^{-\beta E(\mu)},
\]
where $\Lambda$ denotes the set of all possible (eigen-)states of $\HRabi{\e}$. Then, the  Hurwitz-type spectral zeta function is given by
\begin{align}
  \label{eq:speczetamellin}
   \ZARabi{\e}(s;\tau) := \sum_{j=1}^\infty (\lambda^{(\e)}_j +\tau)^{-s} = \frac1{\Gamma(s)}\int_0^\infty t^{s-1} \PARabi{\e}(t;g,\Delta) e^{-t\tau}dt,
\end{align}
where $\lambda^{(\e)}_i$ are the (ordered) eigenvalues in the spectrum of $\HRabi{\e}$. Note that the multiplicity of each eigenvalue is less than or equal to $2$ (see \cite{KRW2017}). 

The main result of this paper is that after setting $\tau=g^2+N$ for some sufficiently large integer $N\in \Z_{\ge0}$ (with respect to $\Delta$ and $|\e|$), we have
\begin{align*}
  \lim_{g \to\infty}  \ZARabi{\e}(s;\tau) &=  \lim_{g \to\infty} \frac1{\Gamma(s)}\int_0^\infty t^{s-1} \PARabi{\e}(t;g,\Delta) e^{-t\tau}dt \\
                                           &= \zeta(s,N + \e) + \zeta(s,N-\e),
\end{align*}
where $\zeta(s,a)$ is the Hurwitz zeta function. Particularly, under certain physically reasonable conditions on $\Delta$, we have
\[
  \lim_{g \to\infty} \ZARabi{0}(s;g^2+N) = 2\zeta(s),
\]
where $\zeta(s)$ is the Riemann zeta function.

\begin{figure}[ht]
  \centering
  \subfloat[$\HRabi{0}$]{
    \includegraphics[height=4.5cm]{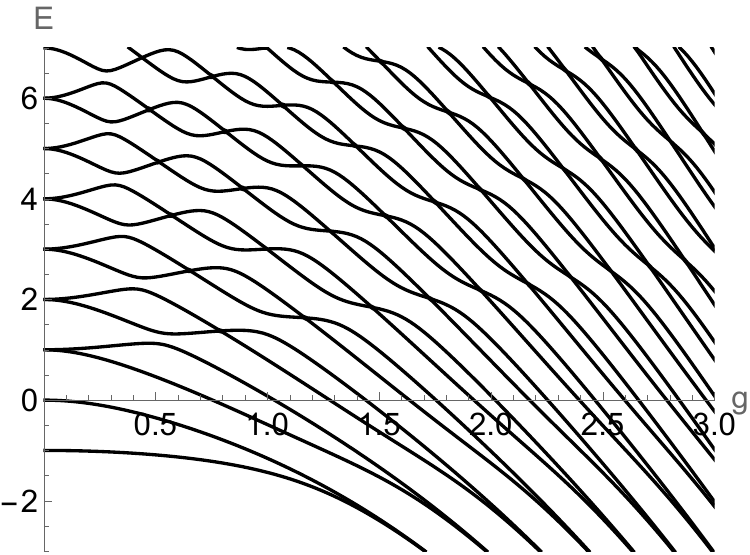}}
  ~ \qquad \qquad
  \subfloat[$\HRabi{0} + g^2$]{
    \includegraphics[height=4.5cm]{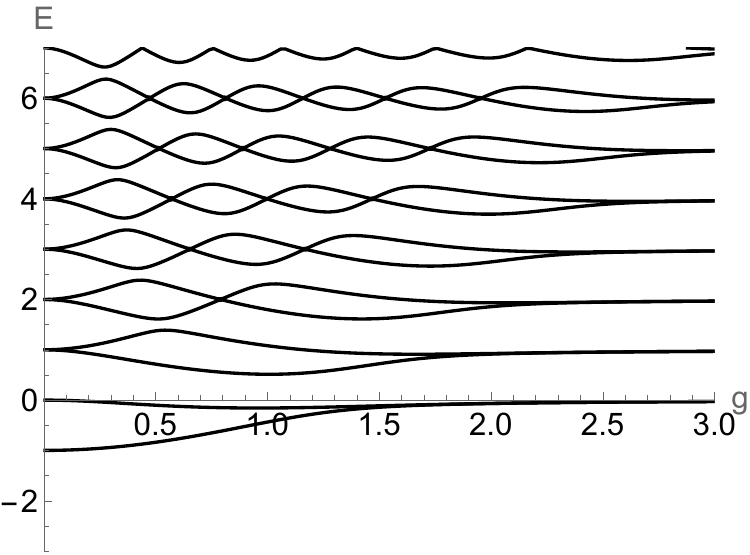}}
  \caption{Spectral curves for $\e = 0$ and $\Delta=1$}
  \label{fig:curve1}
\end{figure}

We note that the displacement $\tau=g^2+N$ is actually necessary for the convergence, as illustrated numerically in the spectral curves in Figure \ref{fig:curve1} for the case of the QRM. It also becomes apparent in the proof of the main result of this paper in Section \ref{sec:spzlinfty}. 

\section{The partition function of the AQRM}
\label{sec:partition}

Let us first recall the formula of the partition function of the AQRM (see \cite{R2020} and \cite{RW2019,RW2021} for the case of the QRM). The formulas for the heat kernel and the partition function are expressed by uniformly convergent series running over $\beta \cdot \Delta$, that is, the time weighted by the parameter $\Delta$.

The infinite series in the partition function may also be interpreted as a type of discrete path integral (see Figure 1 in \cite{RW2019}). Here, the discrete paths appear as representatives of the orbits of the action of the infinite symmetric group $\mathfrak{S}_\infty$ on $\mathbb{Z}_2^{\infty}$ (cf. Appendix C in \cite{RW2021}).

\begin{prop}
  The partition function \( \PARabi{\e}(\beta) (= \PARabi{\e}(\beta;g,\Delta))\) of the AQRM is given by the uniformly convergent series
  \begin{align*}
    \PARabi{\e}(\beta) = \frac{2 e^{g^2\beta}}{1-e^{- \beta}} \Bigg[ &\ch(\e \beta) \\
     &+ \sum_{\lambda=1}^{\infty} (\beta \Delta)^{2\lambda}\idotsint\limits_{0\leq \mu_1 \leq \cdots \leq \mu_{2 \lambda} \leq 1} \Theta_{2 \lambda}(g,\beta,\bm{\mu_{2 \lambda}}) \ch\left[\e \beta \left( 1-2 \sum_{\gamma=1}^{2\lambda} (-1)^{\gamma} \mu_{\gamma}\right) \right] d \bm{\mu_{2 \lambda}}\Bigg],
  \end{align*}
  where
  \[
    \Theta_{2 \lambda}(g,\beta,\bm{\mu_{2 \lambda}}) :=  \exp\left(-2g^2 \coth(\tfrac{\beta}2)+ 4g^2\frac{\ch(\beta(1-\mu_{2\lambda}))}{\sh(\beta)} +  \xi_{2 \lambda}(\bm{\mu_{2\lambda}},\beta) +\psi_{2\lambda}^{-}(\bm{\mu_{2 \lambda}},\beta)\right)
  \]
  with
\begin{align*} 
  \xi_\lambda(\bm{\mu_{\lambda}},t) &:=  -\frac{8g^2 }{\sh(t)} \sh\left(\tfrac12t(1-\mu_\lambda)\right)^2 (-1)^{\lambda}  \sum_{\gamma=0}^{\lambda} (-1)^{\gamma} \ch(t \mu_{\gamma})  \\
                    &\quad  - \frac{4 g^2 }{\sh(t)} \sum_{\substack{0\leq\alpha<\beta\leq \lambda-1\\ \beta - \alpha \equiv 1 \pmod{2}  }}  \left( \ch(t(1-\mu_{\beta+1}))-\ch(t(1-\mu_{\beta})) \right) ( \ch(t  \mu_{\alpha}) - \ch(t \mu_{\alpha+1})), \nonumber \\
  \psi_\lambda^{-}(\bm{\mu_{\lambda}},t) &:=  \frac{4 g^2 }{\sh(t)}\left[ \sum_{\gamma=0}^{\lambda} (-1)^{\gamma} \sh\left(t\left(\tfrac12 - \mu_{\gamma}\right)\right)\right]^2\\
\end{align*}
for \(\lambda \geq 1\) and \(\bm{\mu_{\lambda}} = (\mu_1,\mu_2,\cdots,\mu_\lambda) \) and where \( \mu_0 = 0 \).
\end{prop}

We remark here that the contribution of the flip parameter $\e$ in the formula of  the partition function is limited to $\ch(\e \beta)$ and  $\ch\left[\e \beta \left( 1-2 \sum_{\gamma=1}^{2\lambda} (-1)^{\gamma} \mu_{\gamma}\right) \right]$ in the integrand. 


\begin{ex}
  \label{ex:part0}
  We note that the partition function formula above is valid even for the case $g=0$. In particular, in this case we have
  \begin{align*}
    \PARabi{\e}(\beta) = \frac{2}{1-e^{- \beta}} \Bigg[ &\ch(\e \beta) \\
     &+ \sum_{\lambda=1}^{\infty} (\beta \Delta)^{2\lambda}\idotsint\limits_{0\leq \mu_1 \leq \cdots \leq \mu_{2 \lambda} \leq 1} \ch\left[\e \beta \left( 1-2 \sum_{\gamma=1}^{2\lambda} (-1)^{\gamma} \mu_{\gamma}\right) \right] d \bm{\mu_{2 \lambda}}\Bigg].
  \end{align*}
  Therefore, the series
   \[
     \ch(\e \beta) + \sum_{\lambda=1}^{\infty} (\beta \Delta)^{2\lambda}\idotsint\limits_{0\leq \mu_1 \leq \cdots \leq \mu_{2 \lambda} \leq 1} \ch\left[\e \beta \left( 1-2 \sum_{\gamma=1}^{2\lambda} (-1)^{\gamma} \mu_{\gamma}\right) \right] d \bm{\mu_{2 \lambda}},
   \]
  converges for $\beta>0$ (see also Example \ref{ex:epsNZ}).  
\end{ex}

Let us now consider some estimates for the functions appearing in $\PARabi{\e}(\beta)$ that are used in the proof of the main result in Section \ref{sec:spzlinfty}. First, let us recall the elementary
identities
\begin{equation}
  \label{eq:ident1}
  \coth(\tfrac{t}2) = \frac{ \ch(t)+1}{\sh(t)},\qquad \sh^2(\tfrac{t}2) = \frac{\ch(t) - 1}2 .
\end{equation}
For the function $\psi_{2\lambda}^{-}(\bm{\mu_{\lambda}},t)$, since $\mu_{\gamma}\leq \mu_{\gamma+1}$ we see that
\begin{align*}
  \sum_{\gamma=0}^{2\lambda} (-1)^{\gamma} \sh\left(t\left(\tfrac12 - \mu_{\gamma}\right)\right) &= \sh(\tfrac{t}2) +
 \sum_{\gamma=1}^{\lambda} \left(  \sh\left(t\left(\tfrac12 - \mu_{2\gamma}\right)\right) -  \sh\left(t\left(\tfrac12 - \mu_{2\gamma-1}\right)\right) \right) \\
                                                       & \leq \sh(\tfrac{t}2).
\end{align*}
Therefore, using \eqref{eq:ident1} we obtain
\[
  \psi_\lambda^{-}(\bm{\mu_{\lambda}},t) \leq \frac{4 g^2 \sh^2(\tfrac{t}2) }{\sh(t)}  = \frac{2 g^2 \left( \ch(t)-1\right) }{\sh(t)} .
\]
Similarly, for $\xi_\lambda(\bm{\mu_{\lambda}},t)$, we have
\begin{align*}
    \sum_{\gamma=0}^{\lambda} (-1)^{\gamma} \ch(t \mu_{\gamma}) &= 1 +
    \sum_{\gamma=1}^{\lambda} \left( \ch\left(\tfrac{t}2 \mu_{2\gamma}\right) -  \ch\left(\tfrac{t}2 \mu_{2\gamma-1}\right) \right)\\
                                      &\geq 1,
\end{align*}
and
\begin{align*}
  \ch(t(1-\mu_{\gamma+1})) - \ch(t(1-\mu_{\gamma})) &\leq 0, \\
  \ch(t \mu_{\gamma}) - \ch(t \mu_{\gamma+1}) & \leq 0,
\end{align*}
for $\gamma\geq 0$. It follows that
\begin{align*}
  \xi_{2 \lambda}(\bm{\mu_{\lambda}},t) &\leq -\frac{8g^2 }{\sh(t)} \sh^2\left(\tfrac12t(1-\mu_{2\lambda})\right) =
                          -\frac{4g^2 }{\sh(t)} \left( \ch\left(t(1-\mu_{2\lambda})\right)-1 \right).
\end{align*}

\begin{rem}
There are similar series expressions for the heat kernel of the Kondo model in \cite{AYH1970} and Spin-Boson models in \cite{LCDFGZ1987}. Precisely, in the former the authors derive an explicit  matrix element (for the ground state) of the heat kernel.  The latter treats models having infinite degrees of freedom, containing the QRM as a special model (only one freedom). However, its heat kernel formula is an approximation derived by an stochastic integral, and there are infinite changes of sign in the series expansion of the heat kernel. Hence, in neither case an analytical formula for the partition function is explicitly given.
\end{rem}

\section{The zeta limit of the spectrum}
\label{sec:limit}

Let $H$ be a self-adjoint operator acting on a Hilbert space $\mathcal{H}$ with spectrum consisting only of eigenvalues with uniformly bounded multiplicity. The corresponding Hurwitz-type spectral zeta function $\zeta_H(s,\tau)$ is defined by the Dirichlet series
\[
  \zeta_H(s,\tau) = \sum_{\lambda \in \Spec(H)} \frac{1}{(\lambda+\tau)^s}.
\]
Here, we assume that $\zeta_H(s,\tau)$ is absolutely convergent for $\Re(s)>\sigma \in \R$ possibly after imposing certain conditions on the parameter $\tau$.

\begin{ex}
  \label{Example:qho}
  For the quantum harmonic oscillator with Hamiltonian
  \[
    Q= a^{\dagger}a +\frac12
  \]
  the eigenvalues are given by $n+\frac12$ for $n\in \Z_{\geq0}$.
  The spectral zeta function  $\zeta_{Q}(s,\tau)$ is then given by the Hurwitz zeta function $\zeta(s, \tau+\frac12)$.
  In particular, we have
  \[
    \zeta_{Q}(s, 0) =(2^s-1)\zeta(s).
  \]
\end{ex}

Let us suppose now that the operator $H=H(a)$ depends on a parameter $a\in \C$. We are interested in computing limit
\begin{align}
  \label{eq:zetalimit}
  \lim_{a \to A}\zeta_H(s,\tau)
\end{align}
for a given $A \in \C \cup \{\infty\}$. 

\begin{ex}
  The Jaynes-Cummings model \cite{JC1963}, already mentioned in the introduction, provides a good approximation of
  the QRM in certain parameter regimes. A particular point is that this model has a $U(1)$-symmetry that allows the explicit computation of the spectrum in an elementary way (see e.g. \cite{HR2008}). It follows that 
  the spectral zeta function $\zeta_{\text{JC}}(s,\tau)$ of the Jaynes-Cummings model, after a parameter normalization, is
  given by
  \[
    \zeta_{\text{JC}}(s,\tau) =  \sum_{n=0}^{\infty} \frac{1}{(n+\frac12 \pm \sqrt{\Delta^2 + g^2 (n+1)} + \tau )^s},
  \]
  and is absolutely convergent in $\Re(s)>1$ for $\tau \in \R$ large enough. Here, the parameter $g$ and $\Delta$ have the same
  interpretation as in the case of the AQRM. We also note that it was recently proved in \cite{MM2023} that
  $\zeta_{\text{JC}}(s,\tau)$ can be meromorphically continued to the whole complex plane with  a unique simple pole at $s=1$. 

  For the Jaynes-Cummings model, we can actually compute certain limits of the spectral zeta function directly from
  the series expression in the region of absolute convergence. Concretely, we have
  \[
    \lim_{g \to 0} \zeta_{\text{JC}}(s,\tau) = \zeta(s,\tfrac12 + \Delta + \tau) + \zeta(s,\tfrac12 - \Delta + \tau)
  \]
  and
  \[
    \lim_{\Delta \to 0} \zeta_{\text{JC}}(s,\tau) = \sum_{n=0}^{\infty} \frac{1}{(n+\frac12 \pm g \sqrt{n+1} + \tau )^s}.
  \]
  We note that the limit $\Delta \to 0$ of the Jaynes-Cummings spectral zeta function is actually a more complicated function
  than that of the QRM (see Example \ref{ex:D0} below). One explanation is that since the Jaynes-Cummings model is
  the RWA (rotating wave approximation) of the QRM, it partially breaks the natural mathematical structure of the QRM.
\end{ex}

In general, it may not be possible to compute the limit from the defining series expression even in the region of absolute convergence. 
Recall that $\zeta_H(s,\tau)$ is the Mellin transform of the partition function $Z_{H}(t)$ of $H$, that is,
\begin{align}
  \label{eq:zetaMellin}
  \zeta_H(s,\tau) = \frac1{\Gamma(s)}\int_0^\infty t^{s-1}Z_{H}(t)e^{-t\tau}dt, \qquad\quad  \Re(s)>\sigma.
\end{align}
Therefore, if an explicit form of the partition function $Z_{H}(t)$ is known, it is possible to directly
compute the limit \eqref{eq:zetalimit}.

We call the resulting limit the {\em zeta limit of the spectrum} (or {\em spectral zeta limit}) and we assume the existence of an explicit formula for $Z_{H}(t)$ that allows its computation.
In particular, the spectral zeta limit \eqref{eq:zetalimit} may be computed as
\begin{align}
  \label{eq:defspeclimit}
    \lim_{a \to A} \zeta_H(s,\tau) & = \lim_{a \to A}  \frac1{\Gamma(s)}\int_0^\infty t^{s-1} Z_{H}(t) e^{-t\tau}dt \nonumber \\
   & = \frac1{\Gamma(s)}\int_0^\infty t^{s-1} \lim_{a \to A} Z_{H}(t) e^{-t\tau}dt, 
\end{align}
if the limit exists. Here, $Z_{H}(t)$ should satisfy certain properties to allow for the interchange of limit and series or integrals. 

   Due to a possible existence of iso-spectral operators, it is worth noting that the existence of zeta limit of
   the spectrum does not imply the convergence of the Hamiltonian in a stronger sense. This is another example of
   the famous Kac's drum problem \cite{Kac1968} and we leave a more detailed discussion for another occasion
   including a comparison of the notion of zeta limit with other modes of convergence for operators.

In this paper we consider only the case of the spectral zeta function of the AQRM. To simplify the discussion on
convergence we make the following assumption on $\tau$ {\em throughout the rest of discussion}.

\begin{assump}
  \label{assu:N}
  For the AQRM, the parameter $\tau$ is given by $\tau = g^2+ N$ with \[
    N = [\Delta+|\e|+1],
  \] where $[x]$ is the integer part of $x \in \R$.
\end{assump}

With this assumption, it is not difficult to verify that
\[
  \tau + \lambda >0
\]
for any eigenvalue of the AQRM, that the Mellin integral expression \eqref{eq:zetaMellin} holds when $\Re(s)>1$ and
that the change of limit and integration in \eqref{eq:defspeclimit} is valid (see e.g. \cite{Sugi2016} for the QRM case).

\subsection{Some basic spectral zeta limits for the AQRM}
\label{sec:bszlAQRM}

The spectral zeta limits of the AQRM corresponding to the vanishing of the parameters $g,\Delta$ can be verified in a direct way.
For these examples, the Hamiltonian exists in the limiting case and we observe that the  zeta limit of the spectrum corresponds
to the spectral zeta function of the limit of the Hamiltonian, as expected.

\begin{ex}
  \label{ex:D0}
  The spectral zeta limit for $\Delta \to 0$ is given by
  \begin{align*}
    \lim_{\Delta\to 0}  \ZARabi{\e}(s; \tau) &=  \frac2{\Gamma(s)}\int_0^\infty \frac{ \beta^{s-1} e^{\beta(g^2-\tau)}}{1-e^{- \beta}} \ch(\e \beta) d\beta\\
                                            &= \zeta(s,\tau+\e-g^2) + \zeta(s,\tau-\e -g^2) \\
                                            &= \zeta(s,N+\e) + \zeta(s,N-\e).
  \end{align*}
  Note that this corresponds with the case
  \[
    \lim_{\Delta \to 0}\HRabi{\e} = a^{\dagger}a + g (a + a^{\dagger}) \sigma_x + \e \sigma_x + \tau,
  \]
  which, after an elementary transformation, is the Hamiltonian of a displaced (matrix-valued) Harmonic oscillator.
\end{ex}

\begin{ex} \label{ex1}
  For $\e =0$, the spectral zeta limit $g \to 0$ is given by
  \begin{align*}
    \lim_{g\to 0} \ZARabi{0}(s;\tau) &=  \frac2{\Gamma(s)}\int_0^\infty d\beta \frac{ \beta^{s-1} e^{-\beta N} }{1-e^{- \beta}} \sum_{\lambda=0}^{\infty} (\beta \Delta)^{2\lambda} \idotsint\limits_{0\leq \mu_1 \leq \cdots \leq \mu_{2 \lambda} \leq 1} d \bm{\mu_{2 \lambda}} \\
                                            &= \frac2{\Gamma(s)}\int_0^\infty d\beta \frac{\beta^{s-1} e^{-\beta N} }{1-e^{- \beta}} \sum_{\lambda=0}^{\infty} \frac{(\beta \Delta)^{2\lambda}}{(2\lambda)!} = \frac2{\Gamma(s)}\int_0^\infty  \frac{\beta^{s-1} e^{-\beta N} \ch(\beta \Delta)}{1-e^{- \beta}}  d\beta  \\
                                            &= \zeta(s,N+\Delta) + \zeta(s,N-\Delta).
  \end{align*}
 In this case, the corresponding limit of the Hamiltonian is
  \[
    \lim_{g \to 0}\HRabi{0} = a^{\dagger}a +  \Delta \sigma_z + N,
  \]
  which realizes a pair of decoupled harmonic oscillators.
  \end{ex}

  \begin{ex}
    \label{ex:epsNZ}
    In general, when $g=0$, the Hamiltonian is given by
    \[
      \HRabi{\e} = a^{\dagger}a +  \Delta \sigma_z + \e \sigma_x.
    \]
    Hence, the diagonalization of the matrix shows that the spectrum is given by $\Z_{\geq0} \pm \sqrt{\Delta^2+\e^2}$.
    From this, it follows that
  \[
    \lim_{g\to 0} \ZARabi{\e}(s; \tau) =  \zeta(s,N+\sqrt{\Delta^2+\e^2}) + \zeta(s,N-\sqrt{\Delta^2+\e^2}).
  \]
  We note that this formula first appeared in \cite{Sugi2016} for the case $\e =0$. A direct verification
  from the limit is highly non-trivial, nevertheless it shows a finer structure of  $\lim_{g\to 0} \ZARabi{\e}(s; \tau)$ from  the viewpoint of zeta
  functions (see Section \ref{sec:epsNZ}).
  \end{ex}

\subsection{The spectral zeta limit for $g \to \infty$}
\label{sec:spzlinfty}

In this section we give the main result of the paper, that is,  the spectral zeta limit of the AQRM as $g \to \infty$.
In this case, in contrast with the discussion in Section \ref{sec:bszlAQRM}, the limit of the Hamiltonian cannot be defined in a naive way (see, e.g. \cite{Sch1985}). Nevertheless, the spectral zeta limit exists and agrees with the results of Hiroshima and Shirai for the weak limit $g \to \infty$ using resolvent analysis \cite{H2023}, and physics literature \cite{Sch1985}.
  
\begin{thm}
  \label{thm:ginf}
  Let $\tau=g^2+ N$ be given as in Assumption \ref{assu:N}. The spectral zeta limit as $g \to \infty$ of $\ZARabi{\e}(s; \tau)$ is
  \[
    \lim_{g\to \infty} \ZARabi{\e}(s; \tau) = \zeta(s,N+\e) + \zeta(s,N-\e).
  \]
In particular, if $\e =0$ and  $0<\Delta<1$, we have
  \[
    \lim_{g\to \infty} \ZARabi{0}(s; g^2+1) = 2 \zeta(s).
  \]

\end{thm}

\begin{proof}

  By the estimates of Section \ref{sec:partition}, the argument of the exponential in $\Theta_{2 \lambda}(g,\beta,\bm{\mu_{2 \lambda}})$ is 
  \begin{align}
    \label{eq:ineq1}
    & -2g^2 \coth(\tfrac{\beta}2) +  4g^2\frac{\ch(\beta(1-\mu_{2\lambda}))}{\sh(\beta)} +  \xi_{2 \lambda}(\bm{\mu_{2\lambda}},\beta) +\psi_{2\lambda}^{-}(\bm{\mu_{2 \lambda}},\beta), \nonumber \\
    &\,\leq  \frac{g^2}{\sh(\beta)}\left( -2\ch(\beta) -2  +  4 \ch(\beta(1-\mu_{2\lambda}))  - 4 \ch\left( \beta(1-\mu_{2\lambda}) \right) + 4
      + 2\ch(\beta) -2 \right) \nonumber \\
    &\,= 0.
  \end{align}

  Using this estimate and taking into account Example \ref{ex:part0}, by the dominated convergence we only need
  to consider the limits
  \[
    \lim_{g\to \infty}  \idotsint\limits_{0\leq \mu_1 \leq \cdots \leq \mu_{2 \lambda} \leq 1} \Theta_{2 \lambda}(g,\beta,\bm{\mu_{2 \lambda}}) d \bm{\mu_{2\lambda}},
  \]
  for $\lambda>0$.
  Note that the only possibility for the equality in \eqref{eq:ineq1} is that $\mu_{2\gamma-1}=\mu_{2\gamma}$ for $\gamma=1,2,\cdots,\lambda$.
  We may thus write the integral in the limit as
  \[
    \idotsint\limits_{0\leq \mu_1 \leq \cdots \leq \mu_{2 \lambda} \leq 1} \exp\left(\frac{g^2}{\sh(\beta)} f_{2\lambda}(g,\bm{\mu_{2\lambda}})\right) \ch\left[\e \left( \beta+ \eta_{2\lambda}(\bm{\mu_{2\lambda}},\beta)\right)\right] d \bm{\mu_{2\lambda}},
    \]
    where for any $g>0$,  $f_\lambda(g,\bm{\mu_{2\lambda}})$ is negative a.e. in the domain, that is, except for the zero-measure set
    \[
      \left\{ (x_1,\cdots,x_{2\lambda}) \subset [0,1]^{2\lambda} \,   \big| \, x_{2\gamma-1} = x_{2\gamma}\, , \, \gamma =1,2,\cdots,\lambda \right\}.
    \]
    Therefore
    \[
      \lim_{g\to \infty} \idotsint\limits_{0\leq \mu_1 \leq \cdots \leq \mu_{2 \lambda} \leq 1} \exp\left(\frac{g^2}{\sh(\beta)} f(g,\bm{\mu_{2\lambda}}) \right) \ch\left[\e \left( \beta+ \eta_{2\lambda}(\bm{\mu_{2\lambda}},\beta)\right)\right] d \bm{\mu_{2\lambda}} = 0.
    \]
   Thus, the limit of the partition function is 
    \[
      \lim_{g\to \infty}\PARabi{\e}(\beta;g,\Delta)e^{-\beta (g^2+N)} = \frac{e^{-(N-\e)\beta}+e^{-(N+\e)\beta}}{1-e^{- \beta}},
    \]
    whence the result follows immediately from the definition \eqref{eq:speczetamellin} and the Mellin transform expression of the Hurwitz zeta function.
  \end{proof}

  \begin{rem}
    As mentioned in the introduction, the choice $\tau=g^2+N$ in Assumption \ref{assu:N} and, in particular, the parameter $N$ depend on the scale of values of $\Delta$ and $|\e|$, and are indispensable in the proof of Theorem \ref{thm:ginf}.
    This choice also meets the known result in \cite{Sch1985} that shows that the actual limit $g \to  \infty$ requires the
    additive renormalization $g^2 (=g^2/\omega)$. The need for the additive normalizing  may also be verified directly
    without numerical computation by computing the scalar product of the eigenstates of the two shifted oscillators
    with the same energy and noting that it goes to zero for $g\to \infty$ for the QRM (the case where $\e=0$). On the other
    hand, for the expression of the heat semigroup of the QRM and the Spin-Boson model described by the Feynman–Kac
    integrals \cite{HH2012, HHL2014}, it is not straightforward to see the effect of the additive renormalization
    above, i.e. the shift $g^2$ cannot be found there. 

\end{rem}

\;

\subsection{A remark on the limit of the heat kernel}
\label{sec:limHeat}

In view of the spectral zeta limit of Section \ref{sec:spzlinfty}, it is natural to consider the effect of the limit $g \to \infty$ in the heat kernel
$\KBRabi{\e}$ of $\HRabi{\e}+g^2$. From the definition, it is easy to see that we obtain
$\KBRabi{\e}$ by multiplying  $e^{-\tau t} = e^{-(g^2+N)t}$ to the heat kernel $\KRabi{\e}$.
Concretely, we have
\begin{align*}    
  \KBRabi{\e}(x,y,t) =  e^{-N t} K(x,y,t)  \sum_{\lambda=0}^{\infty} (t\Delta)^{\lambda} \Phi_{\lambda}(x,y,t),
\end{align*}
where $K(x,y,t)$ is the Mehler kernel (for the harmonic oscillator), and $\Phi_\lambda(x,y,t) = \Phi_\lambda(x,y,t;g,\e)$ is a $2 \times 2$-matrix valued function
defined with an iterated integral (see \cite{R2020} for the full expression). 


Using the estimates of Section \ref{sec:partition} (see also Section \ref{sec:partity limit} below), it is not difficult
to verify that
\[
  \lim_{g \to \infty} \Phi_{\lambda}(x,y,t) = \mathbf{0},
\]
the zero matrix, for $\lambda>0$. In contrast to the partition function, the term corresponding to $\lambda =0$ in the series is given by
\[
  \frac{\exp\left( - \frac{1+e^{-2t}}{2(1-e^{-2t})} (x^2 + y^2) +  \frac{2 e^{-t} x y}{1-e^{-2t}} -2g^2 \tanh(\tfrac{t}2) \right)}{\sqrt{\pi (1-e^{-2t})}}
  \begin{bmatrix}
    \cosh  &  - \sinh  \\
    -\sinh &  \cosh
  \end{bmatrix}
  \left( \sqrt2 g(x+y)\frac{1-e^{-t}}{1+e^{-t}} + \e t \right),
\]
multiplied by $e^{-N t}$ and it follows the limit for $g \to \infty$ is the zero matrix.  Consequently, we verify that
\[
  \lim_{g \to \infty} \KBRabi{\e}(x,y,t) = \mathbf{0},
\]
where the changes of limit and integral and series are allowed by the uniform convergence of the series defining the heat kernel.

The vanishing of the heat kernel may appear surprising in light of the spectral zeta limit of Theorem \ref{thm:ginf}.
However,  it is consistent since the heat kernel must satisfy the heat equation
\[
  \frac{\partial}{\partial t} \KBRabi{\e}(x,y,t)= - \left(\HRabi{\e} + g^2+N \right) \KBRabi{\e}(x,y,t),
\]
but the limit of the Hamiltonian $\HRabi{\e}+g^2$ is not defined at the limit $g \to \infty$ (only the weak-limit). Therefore, the limit of the heat kernel vanishes and the heat equation hold trivially.

\section{Zeta limit of the parity decomposition for the QRM}
\label{sec:partity limit}

In this section we consider spectral zeta limits for the parity decomposition of the QRM (that is, $\e=0$).
Surprisingly, but naturally in the sense of the $\Z_2$-symmetry of the QRM, we show that in addition to Hurwitz zeta
functions, $L$-functions corresponding to the Dirichlet character modulo 2 appear in the limits.

We first recall some of the basic facts of the symmetry for the case of the QRM. The Hamiltonian
$ \HRabi{0} $ of the QRM possesses a \(\Z_2\)-symmetry observed by the existence of a parity operator
\[
  \Pi:= - \sigma_{z} e^{-i \pi a^{\dag} a}
\]
satisfying \([\Pi,\HRabi{0} ]=0\) and having eigenvalues \(\{\pm 1\}\). Consequently, the direct decomposition of the full
space $L^2(\R)\otimes \C^2$ into invariant subspaces (corresponding to the positive and negative parity)  is given by
\[
  L^2(\R)\otimes \C^2= \mathcal{H}_+\oplus \mathcal{H}_-, \qquad \quad \mathcal{H}_\pm \simeq L^2(\R).
\]
The Hamiltonians $H_{\pm}$ acting on each parity subspaces are given by
\[
  H_{\pm} = a^{\dag} a + g ( a + a^{\dag}) \pm \Delta \hat{T},
\]
where $(\hat{T}\psi)(z):= \psi(-z)$, for $\psi \in L^2(\R)$, is the reflection operator acting on \( L^2(\R)\).
According to this parity decomposition,  the heat kernel of the QRM is decomposed (Theorem 4.4 \cite{RW2019}) as 
 \[
   \KBRabi{0}(x,y,t) =  \KBRabip^{+} (x,y,t) \oplus \KBRabip^{-} (x,y,t). 
\]
The partition function \(\PRabi^{\pm}(\beta;g,\Delta) \) of the Hamiltonian \(H_{\pm} \) corresponding to each parity is
computed directly from the corresponding heat kernels.

\begin{prop}{(Corollary 4.1 of \cite{RW2019})}
  \label{cor:parityPart}
  The partition function  $\PRabi^{\pm}(\beta)$ for the parity Hamiltonian \(H_{\pm}\) is given by
  \begin{align*}
    &Z^{\pm}(\beta)=  \frac{ e^{g^2\beta}}{1-e^{- \beta}} \Bigg[ 1 +  \sum_{\lambda =1}^{\infty} (\beta\Delta)^{2 \lambda} \idotsint\limits_{0\leq \mu_1 \leq \cdots \leq \mu_{2 \lambda} \leq 1} \Theta_{2 \lambda}(g,\beta,\bm{\mu_{2 \lambda}})  d \bm{\mu_{2\lambda}}   \Bigg] \\
    &\quad \quad  \mp \frac{ e^{g^2 \beta}}{1+e^{- \beta}} \Bigg[ \sum_{\lambda = 0}^{\infty} (\beta \Delta)^{2\lambda+1} \idotsint\limits_{0\leq \mu_1 \leq \cdots \leq \mu_{2 \lambda+1} \leq 1} \Xi_{2\lambda+1} (g,\beta,\bm{\mu_{2 \lambda+1}}) d \bm{\mu_{2\lambda+1}} \Bigg]\\
  \end{align*}
  with
  \[
    \Xi_{2\lambda+1} (g,\beta,\bm{\mu_{2 \lambda+1}}) = \exp\left(- 2g^2 \tanh(\tfrac{\beta}2) + \xi_{2\lambda+1}(\bm{\mu_{2 \lambda+1}},\beta) +\psi^+_{2 \lambda+1} (\bm{\mu_{2\lambda+1}},\beta) \right),
  \]
  and    
  \begin{equation*}
    \psi_\lambda^{+}(\bm{\mu_{\lambda}},t) :=  \frac{4 g^2 }{\sh(t)}\left[ \sum_{\gamma=0}^{\lambda} (-1)^{\gamma} \ch \left(t\left(\tfrac12 - \mu_{\gamma}\right) \right) \right]^2.
  \end{equation*}\
\end{prop}

The spectral zeta functions $\ZRabi^{\pm}(s;\tau)$ of the parity Hamiltonians $H^{\pm}$ are defined in a straightforward way.
Obviously, we have
\[
\ZARabi{0}(s; \tau) =  \ZRabi^{+}(s; \tau)+  \ZRabi^{-}(s; \tau).
\]

As in the case of the spectral zeta limit of (A)QRM, we take $\tau = g^2+N$ satisfying the conditions of Assumption \ref{assu:N} to
ensure asymptotic and convergence properties of the partition functions and the spectral zeta limit.

Since $\ZRabi^{\pm}(s; \tau)$ is expressed by the Mellin transform of $ Z^{\pm}(\beta; g, \Delta)$ respectively, 
by the same reasoning used in Sections \ref{sec:bszlAQRM}, we have the following spectral result.
As customary, for a Dirichlet character $\chi$, the $L$-function $L_\chi(s,\tau)$ is defined by
\[
  L_\chi(s,\tau) = \sum_{n=0}^\infty \frac{\chi(n)}{(n+\tau)^s}.
\]

\begin{thm}
  \label{thm:partialzetalimit}
  Let $\tau=g^2+N$. Let $\chi$ be the Dirichlet character defined by $\chi(n)=(-1)^n$,
  then we have the following spectral zeta limit for $g \to 0$ and $\infty$.
  
  \begin{enumerate}
  \item When $g\to 0$, we have 
   \[
    \lim_{g\to 0}\ZRabi^{\pm}(s; \tau) = \frac12\left\{\zeta(s, N+\Delta) + \zeta(s,N-\Delta) \mp \left(L_\chi(s,N+\Delta)-L_\chi(s,N-\Delta)\right)\right\}  
 \]
  \item When $g \to \infty$, we have
   \[
    \lim_{g\to \infty}\ZRabi^{\pm}(s; \tau)= \zeta(s,N).
  \]
  \end{enumerate}
\end{thm}

\begin{proof}
In the same way as in Example \ref{ex1}, one has 
 \begin{align*}
     \lim_{g\to 0} \ZRabi^{\pm}(s; \tau)& = \frac1{\Gamma(s)}\int_0^\infty \frac{\beta^{s-1} e^{-N \beta} \ch(\beta \Delta)}{1-e^{- \beta}}  d\beta
     \mp \frac1{\Gamma(s)}\int_0^\infty  \frac{\beta^{s-1} e^{- N \beta} \sh(\beta \Delta)}{1+e^{- \beta}}  d\beta.
       \end{align*}
Here we compute 
\[
\int_0^\infty \frac{\beta^{s-1} e^{-N \beta} \sh(\beta \Delta)}{1+e^{- \beta}}  d\beta
= \sum_{n=0}^\infty (-1)^n\int_0^\infty \beta^{s-1} e^{-(n+N)\beta} \sh(\beta \Delta) d\beta
\]
and, hence it follows that
 \begin{align*}
 \frac2{\Gamma(s)}\int_0^\infty \frac{\beta^{s-1} e^{-N \beta} \sh(\beta \Delta)}{1+e^{- \beta}}  d\beta 
& =  \sum_{n=0}^\infty (-1)^n(n+N+\Delta)^{-s}-\sum_{n=0}^\infty (-1)^n(n+N-\Delta)^{-s}\\
& =L_\chi(s,N+\Delta)-L_\chi(s,N-\Delta).
 \end{align*}
 This proves the first assertion.
 
 For the second, let us proceed as in the case of Theorem \ref{thm:ginf}. First, let us note that for $\lambda\geq0$, it is
 straightforward to verify
 \[
   \xi_{2\lambda+1}(\bm{\mu_{2 \lambda+1}},t) \leq  \frac{8g^2 }{\sh(t)} \sh^2\left(\tfrac12t(1-\mu_{2\lambda})\right) \left(1 - \ch(t \mu_{2\lambda+1}) \right),
 \]
 and
 \[
   \psi_\lambda^{+}(\bm{\mu_{2 \lambda + 1}},t) \leq \frac{4 g^2 }{\sh(t)} \left( \ch(\tfrac{t}2) - \ch(t(\tfrac12 - \mu_{2\lambda+1})) \right)^2.
 \]
 The following identity may also be easily verified by direct computation
 \[
   \ch(\tfrac{t}2) - \ch(t(\tfrac12 - \mu_{2\lambda+1})) = 2 \sh(\tfrac{t \mu_{2\lambda+1}}2) \sh(\tfrac{t}2(1- \mu_{2\lambda+1})),
 \]
 and therefore, using this identity and \eqref{eq:ident1} we see that $\xi_{2\lambda+1}(\bm{\mu_{2 \lambda+1}},t) + \psi_{2\lambda+1}^{+}(\bm{\mu_{2\lambda+1}},t)$ is
 bounded above by
 \begin{align*}
    & \frac{4g^2}{\sh(t)} \left( 2\sh^2\left(\tfrac12 t(1-\mu_{2\lambda+1})\right) \left(1 - \ch(t \mu_{2\lambda+1}) \right) + \left( \ch(\tfrac{t}2) - \ch(t(\tfrac12 - \mu_{2\lambda+1})) \right)^2\right) \\
    & = \frac{8g^2 2\sh^2\left(\tfrac12t(1-\mu_{2\lambda+1})\right) }{\sh(t)} \left( 1 - \ch(t \mu_{2\lambda+1}) + 2\sh^2(\tfrac{t}2(\mu_{2\lambda+1}) \right)\\
    & = \frac{8g^2 2\sh^2\left(\tfrac12t(1-\mu_{2\lambda+1})\right)}{\sh(t)} \left(1 - \ch(t \mu_{2\lambda+1}) + \ch(t \mu_{2\lambda+1}) - 1  \right) \\
                                                      &=0.
 \end{align*}
 Therefore, for $\lambda\geq 0$ we see that
 \[
   \Xi_{2\lambda+1} (g,\beta,\bm{\mu_{2 \lambda+1}}) \leq \exp(- 2g^2 \tanh(\tfrac{\beta}2)),
 \]
 and following the proof of Theorems \ref{thm:ginf} we obtain
  \[
    \lim_{g\to \infty}Z^{\pm}(\beta;g,\Delta)e^{-\beta (g^2+N)} =  \frac{e^{-\beta N}}{1 - e^{- \beta}},
  \]
    whence the result follows. 
\end{proof}

\begin{rem}
  By Proposition \ref{cor:parityPart} and the results of this section, we see that for $\tau=g^2+N$ we have
  \[
    \lim_{g \to \infty} \frac1{\Gamma(s)}\int_0^\infty \beta^{s-1}\left(Z^{-}(\beta) - Z^{+}(\beta)\right)e^{-\beta\tau} d \beta = 0.
  \]
  To observe the contribution of the $\Delta$ part as $g \to \infty$ we may consider a modification of the Mellin transform with
  respect to the bounded measure $\mu(\beta)$ on $[0, \infty)$ given by
  \[
    d \mu(\beta) = e^{2g^2 \tanh(\tfrac{\beta}2)}d\beta.
  \]
  Following the proof of Theorem \ref{thm:ginf}, we see that for $\lambda \geq 0$ the integrals
  that the integrals do not vanish and we obtain
  \[
    \lim_{g \to \infty} \frac1{\Gamma(s)}\int_0^\infty \beta^{s-1}\left(Z^{-}(\beta) - Z^{+}(\beta)\right) e^{-\beta \tau} d\mu(\beta) = 2 \Delta L_{\chi}(s+1,N) + H(s),
  \]
  for a function $H(s)$ holomorphic on $\Re(s)>1$.
\end{rem}


\section{Further remarks and discussion}
\label{sec:discussion}

In this section, we make some additional remarks related to the main results of this paper. First, we
discuss the convergence of the spectral curves and their behavior related to certain conjectures in \cite{RBW2022}.
In addition, we show that the spectral zeta function of the QRM (i.e. $\e=0$) gives an interpolation of the Riemann
zeta function $\zeta(s)$ between $g=0$ and the limit  $g \to \infty$. Finally, we show that the use of the
explicit form of the partition function allows new expressions of the spectral zeta function as infinite sums of Hurwitz zeta functions (or other Dirichlet $L$-series) by considering the limit $g \to 0$ as an example.

\subsection{Spectral curves and convergence}
\label{sec:spectralcurve}

As mentioned in the introduction, the behaviour of the spectrum at the limit $g \to \infty$ may be observed numerically in the spectral curves of the AQRM. In particular, in this limit the spectral curves (with respect to $g$  and the energy $E$) converge to the baselines $y= N \pm \e$. Numerically, at least for a small number of spectral curves the convergence appears to be rather fast. We also notice that according to the value of $\Delta$, the spectral curves may cross the baselines several times (we observe that the number of crossings is $O(N)$) before the convergence as illustrated in Figure \ref{fig:curve3}. 

\begin{figure}[ht]
  \centering
  \subfloat[$\Delta=1$]{
    \includegraphics[height=4.5cm]{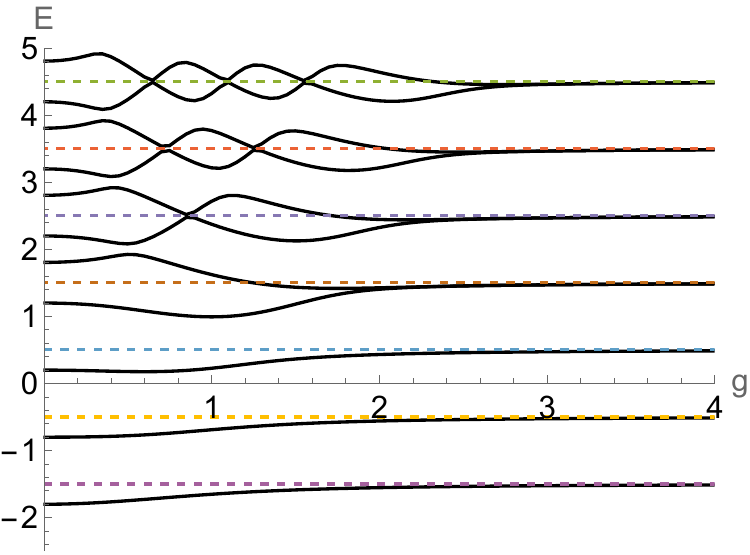}}
  ~ \qquad \qquad
  \subfloat[$\Delta=2$]{
    \includegraphics[height=4.5cm]{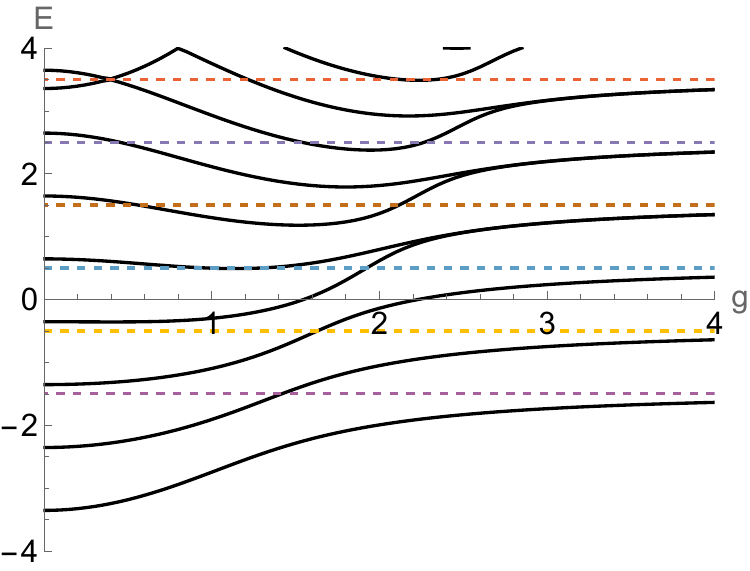}}
  \caption{Spectral curves of $\HRabi{\e}+g^2$ for $\e = 3/2$ for different $\Delta$}
  \label{fig:curve3}
\end{figure}

The results of this paper supports the conjectures of \cite{RBW2022}  regarding the symmetry of the AQRM when
$\e$ is an half integer, that is $\e = \tfrac{\ell}2$ for $\ell \in \Z$. In particular, it has been shown the existence of a commuting operator $J_\ell$ satisfying
\[
  [J_\ell, \HRabi{\tfrac{\ell}2}]=0,
\]
(see \cite{RBW2022} and also \cite{MBB2020} for the discovery of such operators). Let us briefly recall the main conjecture in \cite{RBW2022}. For each half-integer $\e =\frac{\ell}2$, the conjecture proposes the existence of a polynomial $p_\ell(x;g,\Delta) \in \Q[g,\Delta]$ of degree $\ell$ that links the existence of spectral degeneracy and symmetry operators, and such that the associated curve (or surface)
\begin{align}
  \label{eq:helliptic}
  y^2 = p_\ell(x;g,\Delta),
\end{align}
controls the joint spectrum of $H_\e$ and the commuting operator $J_\ell$. More precisely, we have 
\[
  J_\ell^2 = p_\ell(H_{\frac{\ell}2};g,\Delta),
\]
and the polynomial $p_\ell(x;g,\Delta)$ is essentially given by the quotient of two constraint polynomials corresponding to degenerate eigenvalues of the AQRM \cite{KRW2017}.

The conjecture also shows that the hyperelliptic curve defined by the equation \eqref{eq:helliptic} provides an
excellent approximation of the first $\ell$ eigenvalues of the spectrum as shown in Figure \ref{fig:curve4}.
Following this conjecture, it is expected that the first $\ell$ eigenvalues of the AQRM are non-degenerate. This is consistent with the results of the spectral zeta limit as these correspond to the first $\ell$ summands of the series in the definition of $\zeta(s,N-\tfrac{\ell}{2})$.

Moreover, in \cite{RBW2022} it is noted that if ${\rm Ker}(J_\ell)$ is trivial, then the commuting operator $J_\ell$ can be
normalized to define an involution operator. Note that the statement $\dim {\rm Ker}(J_\ell) = 0$ is equivalent to
\[
  p_\ell(\lambda; g,\Delta) \neq 0,
\]
for all eigenvalues $\lambda \in \Spec(\HRabi{\e})$, which is difficult to verify (especially for the first $\ell$
eigenvalues). The results of this paper (along with the result of  Hiroshima and Shirai on the weak limit) suggest 
that at the limit as $g \to \infty$ of the parity operator $J_\ell$ must be $J_\ell =  e^{-i \pi a^{\dag} a}$ and thus $J_\ell^2 = 1$ and $\dim {\rm Ker}(J_\ell) = 0$, supporting the conjecture in \cite{RBW2022}.

\begin{figure}[ht]
  \centering
  \subfloat[$\e=\tfrac52, \Delta=1$]{
    \includegraphics[height=4.5cm]{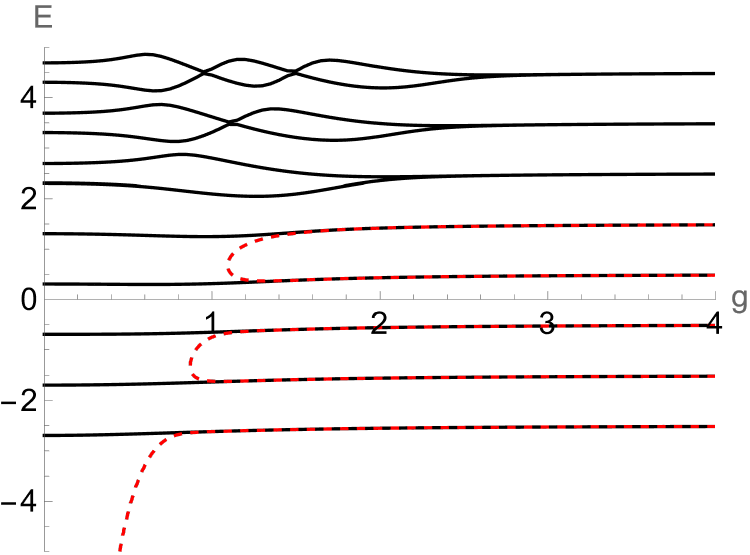}}
  ~ \qquad \qquad
  \subfloat[$\e=\tfrac{10}2, \Delta=2$]{
    \includegraphics[height=4.5cm]{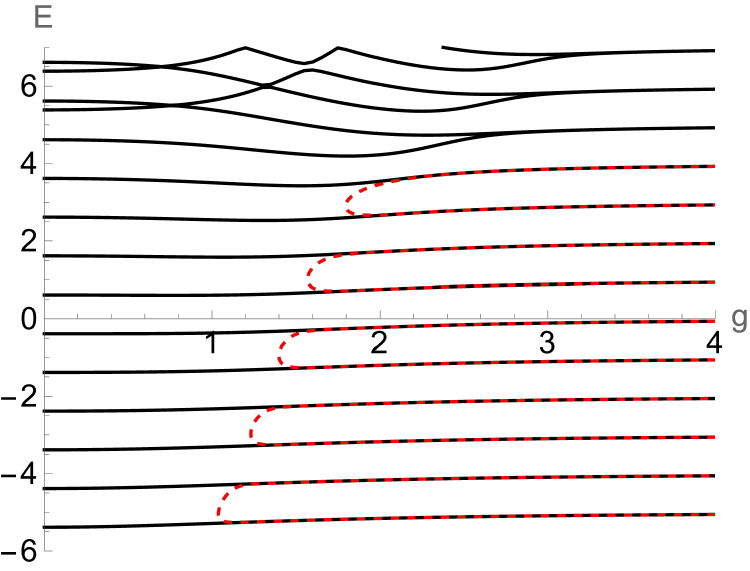}}
  \caption{Excellent approximation of spectral curves}
  \label{fig:curve4}
\end{figure}

\subsection{Spectral zeta interpolation of $\zeta(s)$}
\label{sec:interpolationt}

The Hamiltonian of the QRM may be considered as a (non-commutative) generalization of the quantum harmonic oscillator.
Indeed, by making the parameters $g=\Delta=0$, the Hamiltonian of the QRM is just a pair of uncoupled quantum harmonic oscillators (up to a constant shift).

The same holds for the spectral zeta function $\ZARabi{0}(s; \tau)$, since by Example \ref{ex1} we see that
\[
  \lim_{\substack{ g\to 0 \\ \Delta \to 0 }} \ZARabi{0}(s; g^2+1)=2\zeta(s).
\]
Therefore, it is natural to regard the spectral zeta function of the QRM as a reasonable generalization of the
Riemann zeta function.

Moreover, the the results of this paper show that the spectral zeta function of the QRM \( \ZARabi{0}(s; g^2+1)\),
for fixed $0< \Delta \leq 1$  may be considered as a type of interpolation of the Riemann zeta function with respect to the parameter $g$. In this case, by Theorem \ref{thm:ginf} we actually observe that
\[
  \lim_{g\to \infty} \ZARabi{0}(s; g^2+1)=2 \zeta(s).
\]

From this point of view it is not unreasonable to raise questions such as how the zeros of the zeta function behave under the
interpolation (both trivial and non-trivial zeros) and the perturbation of the Euler product in the intermediate values $g>0$.
A detailed study and clarification of the nature and consequences of the interpolation is beyond the scope of the present paper.
Here, we just limit ourselves to point out an approach that may be considered a first step towards this goal.

It is a classical result that the special value of $\zeta(s)$ at the negative integer $-n$ is given by the Bernoulli
number $B_{n+1}$. Concretely, we have
\[
\zeta(-n)= -\frac{B_{n+1}}{n+1} \quad (n\geq1),
\] 
where $B_k=B_k(0)$, the $k$-th Bernoulli number.

The Rabi-Bernoulli polynomials/numbers from the Appendix A in \cite{RW2021} are defined in the same way to express the special values
of the spectral zeta function of the QRM  (these polynomials are initially introduced in \cite{Sugi2016} by another way of the analytic continuation of $\ZARabi{0}(s;\tau)$).

For $k\geq1$, we define the Rabi-Bernoulli polynomials $(RB)_k(\tau, g^2, \Delta)$ (see Lemma A.3 \cite{RW2019}) by the equation
\[
  \ZARabi{0}(1-k; \tau)=-\frac{2}{k}(RB)_k(\tau, g^2, \Delta).
\]
for $k\geq 0$. Here, for simplicity and to follow the assumptions of this paper, we take $\tau=g^2+1$ and consider
\[
  (RB)_k(g^2, \Delta) := (RB)_k(g^2+1, g^2, \Delta),
\]
thus the first few Rabi-Bernoulli polynomial are given by
\begin{align*}
    (RB)_1(g^2,\Delta) &= B_1,\\
    (RB)_2(g^2,\Delta) &= B_2 + \Delta^2, \\
    (RB)_3(g^2,\Delta) &= B_3+ 3 \Delta^2 B_1+ 2g^2 \Delta^2. \\
\end{align*}
Note that directly from the expression of the Rabi-Bernoulli polynomials, it is not clear how to interpret the interpolation. Let $\Omega(\beta)$ by defined by the equation
\[
  \PARabi{0}(\beta) = \frac{\Omega(t)}{1-e^{-\beta}},
\]
then, more precisely speaking, the $k$-th Rabi-Bernoulli polynomial $(RB)_k(\tau, g^2, \Delta) \in \R[\tau, g^2, \Delta^2] $ is defined
through the equation  
\begin{equation*}
  \frac{w\Omega(w)e^{-(g^2+1) w}}{1-e^{-w}}= 2 \sum_{k=0}^\infty\frac{(-1)^k(RB)_k(g^2, \Delta)}{k!}w^k,
\end{equation*}
and thus it is possible to interpret the limit $g \to \infty$ of the Rabi-Bernoulli polynomials using the methods of this
paper to obtain the usual Bernoulli polynomials.

Moreover, the Rabi-Bernoulli polynomial $(RB)_k(\tau, g^2, \Delta)$ is actually given by a the classical Bernoulli
polynomials. The following results give the explicit formula and allows the computation of the Rabi-Bernoulli
polynomials given above directly. Recall the Bernoulli polynomials \(B_k(x)\) are given by
\[
  \frac{w e^{-x w}}{1-e^{-w}}=  \sum_{k=0}^\infty\frac{(-1)^kB_k(x)}{k!}w^k.
\]

\begin{prop}\cite{RW_RB2023}  \label{prop:Bexp}
 
For \(k \geq 0 \), we have
 \[
    (RB)_k(g^2,\Delta) = (-1)^k k! \left( \sum_{i=0}^k \frac{(-1)^{k-i}}{(k-i)!} \Omega^{(i)} B_{k-i}(1) \right), 
  \] 
where $\Omega^{(i)}$ are defined as the coefficients of $\Omega{(\beta)}$ in the expansion  
\[
  \Omega{(\beta)}= 2e^{g^2}\left( \sum_{j=0}^{\infty} \Omega^{(j)} \beta^j \right).  \qed
\]
\end{prop}

From this observation, we find that the special value of spectral zeta function
\(
  \ZARabi{0}(1-k; g^2+1)
\)
is (up to the term $\Omega{(k)}$) described as a linear combination  of $\zeta(1-j)$ for $1\leq j\leq k$. From this fact, we may expect that a ``major part'' of $\ZARabi{0}(s; g^2+1)$ is given by the finite linear combination of shifts of $\zeta(s)$. We leave
the study of this problem in another occasion.

\subsection{Spectral zeta limit for $g\to0$ as an infinite series of zeta functions }
\label{sec:epsNZ}

In Section \ref{sec:interpolationt} we discussed a possible relationship between the spectral zeta limit with (sums) of the zeta functions for the QRM case. Here, we consider the AQRM (that is, $\e \neq 0$) and return to the case already described in Example \ref{ex:epsNZ} and see that in  the limit $ g \to 0$ the spectral zeta function is given by infinite series of Hurwitz zeta functions. Therefore, while the direct evaluation of the zeta limit of the spectrum may be difficult, the use of the Mellin transform integral expression of the spectral zeta function actually reveals new information.

Concretely, we have
\begin{align}
  \label{eq:multiHZ}
  \lim_{g\to 0} \ZARabi{\e}(s; \tau) 
  =& \zeta(s,N+\e)+ \zeta(s,N-\e) \nonumber \\
   & + \sum_{\lambda=1}^\infty  \Delta^{2\lambda}\sum_{0< p<2\lambda}[a_p\zeta(s-p,N+ \e)+b_p\zeta(s-p, N-\e)  
\end{align}
for some rational numbers $a_p, b_p$. To see this, it is enough to evaluate the integral 
\begin{align}\label{Intergral_AQRM}
  \idotsint\limits_{0\leq \mu_1 \leq \cdots \leq \mu_{2 \lambda} \leq 1}  \ch\left[\e \beta \left( 1-2 \sum_{\gamma=1}^{2\lambda} (-1)^{\gamma} \mu_{\gamma}\right) \right]
  d \bm{\mu_{2 \lambda}}.
\end{align}
Note that all the iterated integrals appearing in the expression \eqref{Intergral_AQRM} are evaluated at $1$,
$\mu_i$ (for $i=2\lambda, 2\lambda-1, \ldots, 1$ ), and $0$ at both ends of the interval in each step.

In practice, by defining the $J_\lambda(t)$ by 
\[
  J_\lambda(t):=  \idotsint\limits_{0\leq \mu_1 \leq \cdots \leq \mu_{\lambda} \leq 1}  \exp\left[t \left( 1-2(-1)^\lambda \sum_{\gamma=1}^{\lambda} (-1)^{\gamma} \mu_{\gamma}\right) \right] d \bm{\mu_{\lambda}}, \qquad (\lambda=1,2,3, \ldots), 
\]
it is immediate to see that the integral \eqref{Intergral_AQRM} is given by $\frac12(J_{2\lambda}(\e \beta ) + J_{2\lambda}(-\e \beta))$.

Next, we note by partial integration that 
\[
  J_1(t)= \frac2{2t}\sh(t), \quad J_2(t)=-\frac2{(2t)^2}\sh(t)+ \frac1{2t}e^t= -\frac1{2t}J_1(-t)+\frac1{2t}J_0(t)
\]
upon defining $J_0(t):=e^t$. Similarly, by multiple application of partial integration  (i.e. $\binom{\lambda}{2}$-times) we
obtain
\begin{align*}
  J_\lambda(t)= \sum_{k=1}^{[\frac{\lambda+ 1}2]} p_k^{(\lambda)}(t^{-1})J_{\lambda-(2k-1)}(-t) +\sum_{k=1}^{[\frac{\lambda}2]} q_k^{(\lambda)}(t^{-1})J_{\lambda-2k}(t) 
\end{align*}
for polynomials $p_k^{(\lambda)}(x)$ and $q_k^{(\lambda)}(x)$ with degree less than or equal to $\lambda$. In particular, it is not difficult to see that $q_{\lambda}^{(2\lambda)}(x)=0$. Actually, we verify the expression of $J_\lambda(t)$ above by noting that the sign change occurs when evaluating at $1$ of the interval $[\mu_k,\,1]$ in each step of the integration and then proceeding by induction. 

Then, it follows  from the expression of $J_{2\lambda}(t)$ above that for $\lambda\in\Z_{>0}$ each coefficient of $(\beta\Delta)^{2\lambda}$ in
$\Omega(t)= (1-e^{-\beta}) \PARabi{\e}(\beta)$ is given by sums of the forms $\varphi_1^{\lambda}((\e\beta)^{-1})\sh(\e \beta)$, $\varphi_2^{\lambda}((\e\beta)^{-1})\ch(\e \beta)$ where
$\varphi_j^{\lambda}(x)\,\, (j=1,2)$ are polynomials in $x$ with rational coefficients of degree less than $2\lambda$.

A few examples are
\begin{align*}
  \varphi_1^{1}(x)=\frac12x, \; \varphi_2^{1}(x)=0 & \quad {\text{i.e., the coefficient of}} \; \beta\Delta =  \frac1{2\e\beta}\sh(\e \beta),\\
  \varphi_1^{2}(x)= -\frac1{8}x^3, \; \varphi_2^{2}(x)= \frac18 x^2  &\quad {\text{i.e., the coefficient of}} \; (\beta\Delta)^2 =\frac1{2(2\e\beta)^2}\ch(\e \beta)-\frac1{(2\e\beta)^3}\sh(\e \beta).
\end{align*}
Also, we find that the highest degree term of $\varphi_1^{\lambda}((\e\beta)^{-1})$ is $(-1)^{\lambda-1}(2\e\beta)^{-2\lambda+1}$ and the degree of the polynomial $\varphi_2^{\lambda}(x)$ is $2\lambda-2$. Indeed, the former claim holds since there are $\lambda$-times sign changes arising from the exponent $t(-1)^\gamma \bm{\mu}_\gamma$ of the integrand of $J_{2\lambda}(t)$ by the $(2\lambda-1)$-times iterated integrals.
The latter follows clearly from the fact $q_{\lambda}^{(2\lambda)}(x)=0$. Hence the assertion in \eqref{eq:multiHZ} follows by taking the Mellin transform.

\section*{Acknowledgements}

The authors would like to thank Fumio Hiroshima and Daniel Braak for discussions related to this research during the conference ``Rabi and Spin Boson models'' held in Kyushu University in January 2023.

This work was partially supported by JSPS Grant-in-Aid for Scientific Research (C) No.20K03560, JSPS Grant-in-Aid for Early-Career Scientists No. 24K16941, and CREST JPMJCR2113, Japan.









\begin{flushleft}
  
\bigskip

  Cid Reyes-Bustos \\
  NTT Institute for Fundamental Mathematics,\\
  NTT Communication Science Laboratories, NTT Corporation \\
3-9-11, Midori-cho Musashino-shi, Tokyo, 180-8585, Japan \\
  email: \texttt{cid.reyes@ntt.com, math@cidrb.me}
  \phantom{a}\\\phantom{a}\\

  Masato Wakayama \\
  NTT Institute for Fundamental Mathematics,\\
  NTT Communication Science Laboratories, NTT Corporation \\
  3-9-11, Midori-cho Musashino-shi, Tokyo, 180-8585, Japan \\
  and\\
  Institute for Mathematics for Industry,  Kyushu University, \\
  744 Motooka, Nishi-ku, Fukuoka 819-0395\, Japan\\
  email: \texttt{wakayama@imi.kyushu-u.ac.jp, masato.wakayama.hp@hco.ntt.co.jp}
  
\end{flushleft}

\end{document}